\documentclass[11pt]{article} 

\usepackage[utf8]{inputenc}
\usepackage[T1]{fontenc}
\usepackage[pdftex,left=1in,top=1in,bottom=1in,right=1in]{geometry}

\usepackage{amsmath, amssymb, dsfont, mathrsfs,amsthm}

\usepackage{graphicx}
\usepackage{caption}
\usepackage{bm}
\usepackage{xcolor}
\usepackage[colorlinks=true,citecolor=blue,linkcolor=blue]{hyperref}
\hypersetup{
    colorlinks,
    linkcolor={red!50!black},
    citecolor={blue!50!black},
    urlcolor={blue!80!black}
}
\usepackage[nameinlink, noabbrev, capitalize]{cleveref}
\usepackage{tikz}
\usetikzlibrary{shapes,arrows,positioning}
\usetikzlibrary{decorations.pathmorphing}
\usepackage[linesnumbered,ruled,vlined]{algorithm2e}
\usepackage{ragged2e}
\usepackage{booktabs}

\crefname{algocf}{Algorithm}{Algorithms}

\newcommand{\N}{\mathbb{N}}

\newcommand{\cS}{\mathcal{S}}
\newcommand{\cT}{\mathcal{T}}

\newcommand{\cX}{\mathcal{X}}
\newcommand{\cY}{\mathcal{Y}}

\newcommand{\cO}{\mathcal{O}}
\newcommand{\cR}{\mathcal{R}}

\newcommand{\NextPos}{\mathbf{NextPos}}

\newcommand{\Cnt}{\mathbf{Count}}

\newcommand{\rem}{\mathrm{rem}}
\newcommand{\subseq}{\mathrm{subseq}}

\newcommand{\cpl}{\operatorname{cpl}}
\newcommand{\rev}{\operatorname{rev}}

\newcommand{\unrankk}{\operatorname{stringUnrank}}
\newcommand{\recsuper}{\operatorname{RecSuperSeq}}
\newcommand{\subunrank}{\operatorname{subsetUnrank}}

\newtheorem{theorem}{Theorem}

\newtheorem{lemma}{Lemma}

\newcommand{\bits}{\{0,1\}}

\newcommand{\BDC}{\mathrm{BDC}}

\newcommand{\fCap}{C}

\newcommand{\rep}{\operatorname{rep}}

\newcommand{\Stab}{\mathrm{Stab}}

\newcommand{\jnote}[1]{{\color{red} \footnotesize(João: #1)}}

\let\originalleft\left
\let\originalright\right
\renewcommand{\left}{\mathopen{}\mathclose\bgroup\originalleft}
\renewcommand{\right}{\aftergroup\egroup\originalright}

\title{Improved Capacity Upper Bounds for the Deletion Channel using a Parallelized Blahut-Arimoto Algorithm
}

\allowdisplaybreaks

\author{Martim Pinto\thanks{Departamento de Matemática, Instituto Superior Técnico, Universidade de Lisboa.
\texttt{martim.velasco.santos.pinto@tecnico.ulisboa.pt}} \and João Ribeiro\thanks{Instituto de Telecomunicações and Departamento de Matemática, Instituto Superior Técnico, Universidade de Lisboa. \texttt{jribeiro@tecnico.ulisboa.pt}}
}

\date{}

\begin{document}
	
\maketitle
	
\begin{abstract}
We present an optimized implementation of the Blahut-Arimoto algorithm via GPU parallelization,
which we use to obtain improved upper bounds on the capacity of the binary deletion channel.
In particular, our results imply that the capacity of the binary deletion channel with deletion probability $d$ is at most $0.3578(1-d)$ for all $d\geq 0.64$.

\end{abstract}

\newpage

\tableofcontents

\newpage

\section{Introduction}\label{sec:intro}

Synchronization errors, such as deletions and insertions, are a common occurrence in communication and data storage systems, most notably in emerging DNA-based data-storage technologies~\cite{YGM17,OAC+18,HMG19,WGGH23}.
This motivates the study of channels with synchronization errors, also called \emph{synchronization channels}.
One of the simplest synchronization channels is the \emph{binary deletion channel} (BDC), which on input a bitstring $x\in\bits^n$ independently deletes each bit of $x$ with some fixed deletion probability $d$.
The corresponding output of this channel would then be the subsequence of $x$ consisting of its ``undeleted'' bits.

The BDC is closely related to the binary erasure channel (BEC), the only difference being that we do not replace the deleted bits by a ``?''.
However, despite this similarity, our state of knowledge about these two channels is wildly different.
We have known the capacity of the BEC since Shannon's early seminal work~\cite{Sha48} and the study of efficient coding for this channel has led to a rich mathematical theory.
On the other hand, we still only know relatively loose bounds on the capacity of the BDC (let alone insights on its other properties), despite an extensive research effort on deriving both capacity lower bounds~\cite{Gal61,VD68,Zig69,DG06,MD06,DM07,DK07,Mit08,KD10,MTL12,RA13,VTR13,CK15,ISW16,RC23} and capacity upper bounds~\cite{DMP07,Mit08,FD10,Dal11,MTL12,RD15,Che19,CR19b,RC23} for the BDC and related synchronization channels.
The main reason behind this is that, although the behavior of the BDC is memoryless like the BEC, it causes a loss of synchronization between sender and receiver: when the receiver looks at the $i$-th output bit, they are not sure to which input bit it corresponds.
For a more detailed discussion of the challenges imposed by this loss of synchronization, see the surveys~\cite{Mit09,MBT10,CR20,HS21}.

\subsection{State-of-the-art capacity upper bounds for the BDC and the underlying barriers}
From here onward we denote the BDC with deletion probability $d$ by $\BDC_d$, and its capacity by $C(d)$. 
The best known upper bounds on $C(d)$ are obtained by numerically approximating the capacity of a ``finite-length'' version of the $\BDC_d$, an approach first studied by Fertonani and Duman~\cite{FD10}.
More precisely, for any given integer $n\geq 1$ we may consider the discrete memoryless channel (DMC) with input alphabet $\bits^n$ and output alphabet $\bits^{\leq n}=\bigcup_{i=0}^n \bits^i$ which on input $x\in\bits^n$ behaves exactly like the $\BDC_d$ on input $x$.
By the noisy channel coding theorem, the capacity of this channel, which we denote by $\fCap_n(d)$, is given by
\begin{equation*}
    \fCap_n(d) = \sup_{X^n}I(X^n;Y),
\end{equation*}
where the supremum is over all random variables $X^n$ supported on $\bits^n$, $Y$ is the corresponding output distribution of the $\BDC_d$ on input $X^n$, and $I(\cdot;\cdot)$ denotes mutual information.
A simple argument (found, for example, in~\cite{FD10,Dal11}) combining the subadditivity of the sequence $\left(\frac{1}{n}\fCap_n(d)\right)_{n\geq 1}$, Fekete's lemma, and the fact, established by Dobrushin~\cite{Dob67}, that
\begin{equation*}
    \lim_{n\to\infty}\frac{1}{n}\fCap_n(d) = C(d),
\end{equation*}
implies that
\begin{equation}\label{eq:fCap-UB}
    C(d)\leq \frac{1}{n}\fCap_n(d)
\end{equation}
for all $n\geq 1$.
Therefore, we can upper bound $C(d)$ by upper bounding $\fCap_n(d)$ for any $n\geq 1$.
For example, by taking $n=1$ we recover the easy upper bound $C(d)\leq 1-d$, valid for all $d\in[0,1]$, and we can obtain better upper bounds by considering larger $n$.

A key observation is that $\fCap_n(d)$ is the capacity of a DMC with finite input alphabet (of size $2^n$) and finite output alphabet (of size $2^{n+1}-1$).
The well-known Blahut-Arimoto algorithm~\cite{Ari72,Bla72} can, in principle, numerically approximate the capacity of any DMC with finite input and output alphabets to any desired accuracy.
By the connection above, arbitrarily good numerical approximations of $C_n(d)$ would lead to arbitrarily good approximations of $C(d)$, for any deletion probability $d$.\footnote{Fertonani and Duman~\cite{FD10} and later works, including ours, actually numerically approximate the capacity of the related \emph{exact deletion channels} parameterized by $n\geq 1$ and $0\leq k\leq n$ which receive an $n$-bit string $x$ and delete a uniformly random subset of $n-k$ bits of $x$, and then use these values to upper bound $C(d)$.
This discussion applies equally well to those channels.
For the sake of simplicity, we focus on a direct analysis of $\fCap_n(d)$ here and leave a discussion of these proxy channels to \cref{sec:aux-channel}.}

The main issue with the approach in the previous paragraph is that the time and space complexity of the Blahut-Arimoto scales badly with the size of the input and output alphabets of the DMC under analysis.
Therefore, a naive implementation of the Blahut-Arimoto algorithm will only produce results in a reasonable timeframe for small values of the input length $n$. 
For example, Fertonani and Duman~\cite{FD10} were able to run the BAA only up to $n=17$.
This motivates the following challenge: 
\begin{quote}
    \em
    Can we optimize the implementation of the Blahut-Arimoto with the deletion channel in mind so that we can obtain good bounds on $\fCap_n(d)$ for significantly larger $n$?
\end{quote}

\noindent Recently, Rubinstein and Con~\cite{RC23} took on this challenge and presented an implementation of the Blahut-Arimoto algorithm with lower space complexity for this problem.
They were able to apply this algorithm to input lengths up to $n=28$, obtaining the current state-of-the-art upper bounds on $C(d)$.

\subsection{Our contributions}

We present an optimized implementation of the Blahut-Arimoto algorithm using GPU parallelization, and use it to obtain improved upper bounds on the capacity of the BDC for the entire range of the deletion probability $d$.

More precisely, we use our optimized implementation of the Blahut-Arimoto algorithm to compute upper bounds on $\fCap_n(d)$ for input lengths up to $n=31$.\footnote{To be more precise, we computed good approximations of the exact deletion channel capacities $C_{n,k}$ (see \cref{sec:aux-channel}) for all $k\leq n\leq 29$, and for $n=31$ and all $k\leq 18$. In contrast, Rubinstein and Con~\cite{RC23} were able to compute good approximations of the $C_{n,k}$ values for all $n\leq 28$ and all $k$ satisfying $k + n \leq 39$. They were not able to approximate $C_{n,k}$ for some values of $k$ when $22 \leq n \leq 28$, due to the high space and time complexity of their implementation.}
The resulting improved bounds on $C(d)$ are reported in \cref{tab:C_d} for several values of $d$.
Here, we expand on their consequences in the asymptotic high-noise setting where $d\to 1$, also studied in prior works~\cite{MD06,DMP07,FD10,Dal11,RD15,RC23}.
Combining our improved bound for $d=0.64$ with a result of Rahmati and Duman~\cite{RD15}, we conclude that
\begin{equation*}
    C(d)\leq 0.3578(1-d)
\end{equation*}
for all $d\geq 0.64$.
This improves on the previous best bound in the high-noise regime due to Rubinstein and Con~\cite{RC23}, which was $C(d)\leq 0.3745(1-d)$ for all $d\geq 0.68$.

The implementation used to obtain the upper bounds is publicly available in a GitHub repository~\cite{repo}.

\subsection{Acknowledgements}

We thank Roni Con for several insightful discussions that improved this paper.

This work was funded by the European Union (LESYNCH, 101218842).
Views and opinions expressed are however those of the authors only and do not necessarily reflect those of the European Union or the European Research Council Executive Agency. Neither the European Union nor the granting authority can be held responsible for them.

\section{Preliminaries}
\label{sec:prelims}

\subsection{Notation}

We denote random variables and sets by uppercase Roman letters such as $X$, $Y$, and $Z$.
Sets are sometimes also denoted by uppercase calligraphic letters such as $\cS$ and $\cT$.
For an integer $n$, we define $\bits^{\leq n}=\bigcup_{i=0}^n \bits^i$.
We index strings starting at $0$, and for $y\in\bits^n$ we define $y_{[a:b]}=(y_a,y_{a+1},\dots,y_b)$.
We write $\log$ for the base-$2$ logarithm.

\subsection{Exact finite-length deletion channels}\label{sec:aux-channel}

As already discussed in \cref{sec:intro}, our starting point is a finite-length version of the $\BDC_d$.
More precisely, given an integer $n\geq 1$ this is the DMC that accepts inputs from $\bits^n$ and, given $x\in\bits^n$, sends $x$ through the $\BDC_d$ and returns its output $y\in\bits^{\leq n}$.
We denote the capacity of this DMC by $\fCap_n(d)$.
As mentioned before, the following inequality holds.
\begin{lemma}[\protect{\cite{FD10,Dal11}}]
    For every $d\in[0,1]$ and integer $n\geq 1$ we have
    \begin{equation*}
        C(d)\leq \frac{1}{n}\fCap_n(d).
    \end{equation*}
\end{lemma}

Since this finite-length channel is a DMC, its capacity $\fCap_n(d)$ can be numerically approximated in principle using the Blahut-Arimoto algorithm, provided sufficient computational resources.
A disadvantage of this approach is that we would need, at least at first sight, to restart the computation from scratch if we change the deletion probability $d$.
With this in mind, it is also useful to consider another family of finite-length versions of the BDC, called \emph{exact deletion channels}.

An \emph{exact deletion channel} is parameterized by an input length $n\geq 1$ and another integer $0\leq k\leq n$.
We denote this channel by $\BDC_{n,k}$.
On input $x\in\bits^n$, $\BDC_{n,k}$ outputs a length-$k$ subsequence $y$ of $x$ uniformly at random over all such $\binom{n}{k}$ subsequences.
In other words, $\BDC_{n,k}$ chooses a uniformly random subset of $n-k$ coordinates of $x$ and deletes those bits.
This channel is a DMC with input alphabet $\cX=\bits^n$, output alphabet $\cY=\bits^k$, and channel rule $P_{n,k}$ satisfying
\begin{equation*}
    P_{n,k}(y|x)=\frac{\# \textrm{times $y$ appears as a subsequence of $x$}}{\binom{n}{k}}.
\end{equation*}
Its capacity, denoted $C_{n,k}$, is thus given by
\begin{equation*}
    C_{n,k} = \sup_{X^n} I(X^n; Y),
\end{equation*}
where the supremum is over all distributions $X^n$ supported on $\bits^n$ and $Y$ is the corresponding channel output on input $X^n$.
The values of $C_{n,k}$ for $1\leq k\leq n$ can be used to bound $C_n(d)$ by taking an appropriate convex combination.
\begin{lemma}[\protect{\cite{FD10}}]\label{lem:prob-vs-exact}
    For every $d\in[0,1]$ and all integers $n\geq 1$ and $1\leq k\leq n$ we have
    \begin{equation*}
        C_n(d) \leq \sum_{k=1}^n \binom{n}{k}d^{n-k}(1-d)^kC_{n,k}.
    \end{equation*}
\end{lemma}
An advantage of this approach is that once we upper bound $C_{n,k}$ for all $1\leq k\leq n$, we can then easily get upper bounds on $C_n(d)$ for any value of $d$.
We will follow this approach in this work.

Rubinstein and Con~\cite{RC23} used their optimized implementation of the Blahut-Arimoto algorithm to approximate $C_{n,k}$ for $k + n \leq 39$ with $n\leq 28$. However, they were not able to compute $C_{n,k}$ for some values of $k$ when $22 \leq n \leq 28$, due to the high space and time complexity. In the following sections, we will see how to further decrease the time and space complexity of the algorithm specifically for deletion channels.

\subsection{The Blahut-Arimoto algorithm}\label{sec:baa}

The Blahut-Arimoto algorithm (BAA) is a well-known tool for numerically approximating the capacity of finite-input/finite-output DMCs~\cite{Ari72,Bla72}. 
We present the standard formulation of the BAA here, and later show how it can be optimized for the BDC.

A finite-input/finite-output DMC is characterized by a finite input alphabet $\mathcal{X}$, a finite output alphabet $\mathcal{Y}$, and the channel law $P$.
For each $x\in\cX$ and $y\in\cY$ the probability that the channel outputs $y$ on input $x$ is denoted by $P(y|x)$.
The BAA is an iterative procedure for approximating the capacity of any such DMC.
After a prescribed number of iterations, it returns an input distribution $X$.
The corresponding achievable rate $I(X;Y)$ (here $Y$ is the channel output distribution given input $X$) of the input distribution $X$ returned by the BAA is guaranteed to converge to the capacity of the channel as the number of iterations increases.
In fact, we have more precise knowledge about the rate of convergence, as stated in the following result.

\begin{theorem}[{\cite{Ari72}}]\label{claim:baa}
Fix an finite-input/finite-output DMC, and let $C$ be its capacity.
For any threshold $a > 0$, the BAA with $O(1/a)$ iterations
returns an input distribution $X$ whose information rate $I(X;Y)$ satisfies
\[
    C-a\leq I(X;Y)\leq C.
\]
Here, the $O(\cdot)$ notation hides a multiplicative constant that depends only on the choice of the DMC.
Furthermore, the convergence to $C$ is monotonic.
\end{theorem}

We now describe the BAA with a conservative stopping criterion.
For a more detailed discussion, see~\cite[Chapter 9]{Yeu08}.
The starting point for the BAA is an initial input distribution $X^{(0)}$.
This distribution may be chosen arbitrarily subject to having full support over $\cX$. 
A common choice is to take $X^{(0)}$ to be the uniform distribution on $\cX$.
Then, for $t\geq 0$, the algorithm proceeds as follows on the $t$-th iteration given $X^{(t)}$:
\begin{enumerate}
    \item \textbf{Compute  channel output distribution of $X^{(t)}$:} This is the distribution $Y^{(t)}$ satisfying
    \begin{equation*}
        Y^{(t)}(y)= \sum_{x\in\cX} X^{(t)}(x) P(y|x)
    \end{equation*}
    for all $y\in\cY$.

    \item \label{step:refined-input}
    
    \textbf{Compute  refined input distribution:}
    This is divided into two steps.
    \begin{enumerate}
        \item We compute the ``unnormalized distribution'' $W^{(t)}$ given by
        \begin{equation*}
            W^{(t)}(x) = \prod_{y\in\cY} \left(\frac{X^{(t)}(x) P(y|x)}{Y^{(t)}(y)}\right)^{P(y|x)}
        \end{equation*}
        for all $x\in\cX$.

        \item We normalize $W^{(t)}$ to get the new input distribution $X^{(t+1)}$.
        That is, we compute
        \begin{equation*}
            X^{(t+1)}(x)= \frac{W^{(t)}(x)}{\sum_{x'\in\cX}W^{(t)}(x')}
        \end{equation*}
        for all $x\in\cX$.
    \end{enumerate}

    \item \textbf{Stopping criterion:} Suppose that we wish to output an input distribution $X$ such that its information rate $I(X;Y)$ satisfies $I(X;Y)\geq C-a$, with $C$ the capacity of the DMC and $a>0$ some approximation threshold.
    Then (e.g., see~\cite[Equation (32)]{Ari72}), it suffices to check whether
    \begin{equation*}
        \max_{x\in\cX} \log\left(\frac{X^{(t+1)}(x)}{X^{(t)}(x)}\right) < a.
    \end{equation*}
    If this condition is satisfied, we stop and return $X=X^{(t+1)}$.
    Otherwise, we move to the next iteration of the BAA.
    
\end{enumerate}

\subsubsection{The optimizations of Rubinstein and Con}\label{sec:rc}

A naive implementation of the BAA requires (i) storing the whole $|\cX|\times|\cY|$ channel transition matrix (that for each $x\in\cX$ and $y\in\cY$ stores $P(y|x)$), and (ii) storing all the values $W^{(t)}(x)$ for $x\in\cX$ in \cref{step:refined-input} of the BAA.
We wish to apply the BAA to numerically approximate the capacities $C_{n,k}$, corresponding to a DMC with input alphabet size $|\cX|=2^n$ and output alphabet size $|\cY|=2^k$.
Therefore, the memory costs of a naive implementation of the BAA quickly become prohibitive as the input length $n$ increases.

Rubinstein and Con~\cite{RC23} developed a more efficient implementation of the BAA for computing the $C_{n,k}$ values through time-memory tradeoffs and by leveraging sparsity of the relevant matrices when $k$ is close to $n$.
Since their methods are also relevant to our optimized implementation of the BAA, we describe them in more detail.
To handle the transition matrix $P$, there are two extremes we can consider:
\begin{itemize}
    \item Pre-compute and store the whole transition matrix (requiring time and space $\Omega(2^{n+k})$). The advantage of this method is that, once this is done, we can retrieve the $P(y|x)$ values in time $O(1)$;

    \item Every time we require $P(y|x)$, compute it from scratch.
    Each such computation can be done in time $\Theta(nk)$ through dynamic programming.
\end{itemize}
Both extremes (high storage/low retrieval costs vs.\ low storage/high retrieval costs) turn out to be too costly even for small values of $n$, due to memory or time constraints.
Rubinstein and Con adopted an approach between the two extremes.
They devise a ``loop-nest'' optimized implementation of the BAA, consisting in cleverly changing the order of operations in the BAA, and combine it with the pre-computation of a smaller table (cache) that can then be used to compute the transition probabilities $P(y|x)$ much faster than the baseline $\Theta(nk)$ time procedure.
To further speed up the application of the BAA, they leverage sparsity when $k$ is close to $n$.

For completeness, \cref{alg:baa} describes the loop-nest optimized implementation of the BAA from \cite{RC23} verbatim with adapted notation.
The method used in \cite{RC23} to balance storage and time requirements for computing the transition probabilities is described verbatim in \cref{alg:pxy}.
We give a more detailed description of how this approach works.
In order to compute the transition probabilities $P_{n,k}(y|x)$ for the $\BDC_{n,k}$, the algorithm leverages only pre-computed tables containing the transition probabilities for the ``smaller'' exact deletion channels $\BDC_{n',k'}$ with $n'\approx n/2$ and $k'\leq k$.
Then, to compute $P_{n,k}(y|x)$ for some input string $x \in \{0,1\}^n$ and output string $y \in \{0,1\}^k$, the algorithm partitions $x$ into two halves, $x_1$ and $x_2$, of lengths $n_1 =\lfloor n/2\rfloor$ and $n_2= \lceil n/2\rceil$, respectively.
Then, it computes $P_{n,k}(y|x)$ based on the pre-computed table by decomposing it as
\begin{equation*}
    P_{n,k}(y|x) = \sum_{k'=0}^k P_{n_1,k'}(y_{[0:k'-1]}|x_1)\cdot P_{n_2,k-k'}(y_{[k':k-1]}|x_2),
\end{equation*}
where $y_{[a:b]}=(y_a,y_{a+1},\dots,y_b)$ and we recall that in this work we index strings starting at $0$, and so write $y=(y_0,y_1,\dots,y_{k-1})$.

This recursive formulation reduces the per-query time complexity to $O(k)$, compared with the naive $O(nk)$ computation, while requiring $O(2^{\frac{n}{2} + k})$ space to store the pre-computed tables. In practice, this trade-off provides a favorable balance between time and space.

\begin{algorithm}
\caption{Loop nest optimized BAA \cite[Algorithm 4 with adapted notation]{RC23}}
\label{alg:baa}
\KwIn{Input/output alphabets $\mathcal{X}, \mathcal{Y}$; channel law $P$; convergence threshold $a > 0$}
\KwOut{Input distribution $X$ with $I(X;Y)=R$}
    $t \gets 0$\;
    Choose $X^{(0)}$ to be the uniform distribution on $\cX$\;

    \Repeat{
        $\displaystyle \max_{x \in \mathcal{X}} \left| \log \left( \frac{X^{(t)}(x)}{X^{(t-1)}(x)} \right) \right| < a$
    }{
        \tcp{Compute output distribution}
        \ForEach{$y \in \mathcal{Y}$}{
            $\displaystyle Y^{(t)}(y) \gets \sum_{x \in \mathcal{X}} X^{(t)}(x) P(y|x)$\;
        }

        \tcp{Compute auxiliary function}
        \ForEach{$x \in \mathcal{X}$}{
            $\displaystyle W^{(t)}(x) \gets \prod_{y\in\cY} \left(\frac{X^{(t)}(x)P(y|x)}{Y^{(t)}(y)}\right)^{P(y|x)}$\;
        }

        \tcp{Update input distribution}
        \ForEach{$x \in \mathcal{X}$}{
            $\displaystyle X^{(t+1)}(x) \gets \frac{W^{(t)}(x)}{\sum_{x'\in\cX}W^{(t)}(x')}$\;
        }

        $t \gets t + 1$\;
    }

    \tcp{Compute information rate}
        $\displaystyle R \gets I(X^{(t)};Y^{(t)})$\;

    \Return{$X^{(t)}, R$}\;
\end{algorithm}

\begin{algorithm}
\caption{Cache-based computation of transition probabilities \cite[Algorithm 3 with adapted notation]{RC23}}
\label{alg:pxy}
\KwIn{Input string $x \in \mathcal{X} = \{0,1\}^n$; output string $y \in \mathcal{Y} = \{0,1\}^k$; cache table of transition probabilities $P_{n',k'}(y'|x')$ for all $n' \leq \lceil n/2 \rceil$ and $k' \leq k$}
\KwOut{Transition probability $P_{n,k}(y|x)$}

    $n_1 \gets \lceil\frac{n}{2}\rceil$\;
    $n_2 \gets \lfloor\frac{n}{2}\rfloor$\;
    
    $x_1 \gets x_{[0:n_1-1]}$
    $x_2 \gets x_{[n_1:n-1]}$ \;
    
    $\displaystyle P_{n,k}(y|x) \gets \sum_{k'=0}^k P_{n_1,k'}(y_{[0:k'-1]}|x_1) \cdot P_{n_2,k-k'}(y_{[k':k-1]}|x_2)$\;

    \Return{$P_{n,k}(y|x)$}\;
\end{algorithm}

\section{An overview of our optimizations}\label{sec:overview}

We provide an optimized implementation of the BAA by leveraging the observation that various steps in an iteration of the BAA lend themselves easily to parallelization.
For example, consider the task of computing the auxiliary function $W^{(t)}(x)$ for all $x\in\cX$ in the BAA applied to the $\BDC_{n,k}$, which we may write equivalently as
\begin{equation*}
    \log W^{(t)}(x)=\sum_{y\in\cY} P_{n,k}(y|x) \log\left(\frac{X^{(t)}(x) P_{n,k}(y|x)}{Y^{(t)}(y)}\right)
\end{equation*}
for increased numerical stability.
First, note that we can compute the values of $ \log W^{(t)}(x)$ for distinct $x$'s in parallel.
Second, for each $x\in\cX$ we can further parallelize the computation by computing each term in the sum over $y\in\cY$ in parallel.
Of course, in practice we only have access to a limited number of parallel threads, but it is not hard to arrange the computations to fit this, for example by assigning to each thread more than one $y$ in the sum.

Motivated by this, we set up parallelized versions of these steps that fit nicely into \emph{CUDA kernels}.\footnote{For an introduction to CUDA, see the CUDA C++ programming guide at https://docs.nvidia.com/cuda/cuda-c-programming-guide/.} 
A CUDA kernel is divided into a grid of blocks, \emph{with each block executing the same function in parallel}.  
Within each block, up to 1024 threads execute concurrently, but this number can also be any smaller power of $2$. 
We now discuss a way of parallelizing this computation using CUDA kernels:
\begin{enumerate}
    \item Assign a distinct input $x\in\cX$ to each block in the kernel;

    \item \label{step:partial-sums} Let $\cS_{x,k}$ denote the set of all length-$k$ subsequences of $x$.
    We partition $\cS_{x,k}$ into up to $1024$ disjoint subsets $\cT_i$ for $i\in\{1,\dots,1024\}$.
    Then, the $i$-th thread of the block corresponding to $x$ computes the partial sum
    \begin{equation*}
        A_i=\sum_{y\in\cT_i}P_{n,k}(y|x) \log\left(\frac{X^{(t)}(x) P_{n,k}(y|x)}{Y^{(t)}(y)}\right),
    \end{equation*}
    where $X^{(t)}(x)$ and $Y^{(t)}(y)$ are already known for all $x\in\cX$ and $y\in\cY$.
    Concretely, if $y_i$ denotes the $i$-th length-$k$ subsequence of $x$ in some pre-specified ordering, then we take $\cT_i=\{i,i+1024,i+2\cdot 1024,\dots\}$.

    \item Compute $\log W^{(t)}(x)=\sum_{i=1}^{1024} A_i$.
\end{enumerate}

Note that our description of our parallelization of the computation of $W^{(t)}(x)$ above is not complete.
We still need to describe how we compute $\cS_i$ and the relevant transition probabilities $P_{n,k}(y|x)$ in \cref{step:partial-sums}.
In other words, we must be able to efficiently enumerate the subsequences of $x$ in $\cS_i$ (to carry out the partial sum), and, for each $y\in\cS_i$, compute the transition probabilities $P_{n,k}(y|x)$.

For computing the transition probabilities $P_{n,k}(y|x)$ we rely on the approach of Rubinstein and Con~\cite{RC23} discussed in \cref{sec:rc} (in particular, see \cref{alg:pxy}).
Namely, we pre-compute the full table of transition probabilities for input lengths $n'\approx n/2$.
Then, each thread in our parallel computation accesses this table to compute $P_{n,k}(y|x)$ for the subsequences $y$ it enumerates over.

For enumerating over the required subsequences, 
a naive approach would be to iterate over all $\binom{n}{k}$ subsets of coordinates of $x$.
However, given the typical memory constraints of GPUs, particularly limited RAM, coupled with the concurrent execution of multiple blocks, it is infeasible to maintain large per-block lookup tables to track which subsequences have already been generated.
Instead, we rely on dynamic programming-based techniques that only require a look-up table of size $O(nk)$ that can also be constructed in time $O(nk)$.
Denote by $N_{x,k}$ the number of length-$k$ subsequences of $x$.
Each thread in the block of the CUDA kernel assigned to input $x$ needs to enumerate over $N\approx \frac{N_{x,k}}{1024}$ length-$k$ subsequences of $x$.
Using our dynamic programming-based method, this can be done in time $O(Nk)$, with small hidden constants.
Since $k\ll 1024$, this yields a significant efficiency improvement over having a single thread enumerate all $N_{x,k}$ subsequences of $x$ in time $\Theta(N_{x,k})$, and so we improve significantly over previous implementations of the BAA.

Other computations in an iteration of the BAA, such as the computation of $Y^{(t)}(y)$ for all $y\in\cY$, can be similarly parallelized.
In this case, we assign each $y\in\cY$ to a different block, and partition the length-$n$ supersequences $x$ of $y$ into up to $1024$ disjoint subsets.
Since the implementation of these ideas is similar to the above, we avoid discussing them further to avoid cluttering the exposition.

We discuss the methods we use to enumerate subsequences and supersequences in more detail in \cref{sec:subseq,sec:superseq}, respectively.
In \cref{app:output-seqs} we discuss an alternative method for computing the channel output probabilities $Y^{(t)}(y)$ via subsequence enumeration, leveraging some simple symmetries of $\BDC_{n,k}$, that is faster for certain values of $k$.

\section{Enumerating subsequences of a given length}
\label{sec:subseq}

In this section, we discuss the method we use to enumerate subsequences.
More precisely, our goal is to, given a string $x\in\bits^n$, a subsequence length $k$, and an integer $j$, return the $j$-th length-$k$ subsequence of $x$ in some arbitrary but pre-specified order.
As discussed in \cref{sec:overview}, our aim is an enumeration algorithm that is well-suited for execution on GPUs, which have limited memory constraints. 
We consider a dynamic programming-based approach that uses a small look-up table.
For completeness, we describe the algorithm and prove its correctness.

Before providing a technical description of the procedure we present the intuition behind it.
The enumeration, or \emph{unranking}, procedure constructs the subsequence $y$ bit-by-bit from left to right.
At each step, we find the earliest possible next occurrence of $0$ and $1$ in $x$ after the last coordinate of $x$ we included in the subsequence. 
All length-$\rem$ completions of $y$ whose next symbol is $0$ are ordered before those whose next symbol is $1$.
We keep track of how many completions lie in each of these two sets, which we call the $0$-set and the $1$-set.
We then compare the current rank $j$ with the size of the $0$-set. If $j$ is at most the size of the $0$-set we append $0$ to $y$; otherwise, we subtract the size of the $0$-set from $j$ and append a $1$ to $y$.
Then, we move on to the next bit of the subsequence.
Repeating this procedure $k$ times constructs the $j$-th length-$k$ subsequence in lexicographic order.

\subsection{Pre-computed tables}

The enumeration is based on two precomputed tables, described below.
In our CUDA implementation of the BAA, each block of the CUDA kernel constructs separate tables (because each block is assigned to a different $x$).
In each block only one thread pre-computes the tables and stores them in dedicated memory.
\begin{itemize}
  \item $\NextPos$: for $0\le i\le n$ and $b\in\{0,1\}$, $\NextPos[i][b]$ stores the smallest index $p\ge i$ with $x_p=b$, or $n$ if no such index exists. 

  The $\NextPos$ table is standard and computed in linear time by scanning from right to left.
  \cref{alg:buildnextpos} describes the procedure we use to construct the $\NextPos$ table.

  \item $\Cnt$:
    for $0\le i\le n$, $0\le t\le k$, $\Cnt[i][t]$ satisfies
    \begin{equation*}
        \Cnt[i][t] =
        \left|\{s\in\{0,1\}^t : s \textrm{ is a subsequence of } x_{[i: n-1]}
        \}\right|.
    \end{equation*}
    We use the boundary conditions $\Cnt[i][0]=1$ for all $i$ (the empty subsequence) and $\Cnt[n][t]=0$ for all $t>0$.

    The $\Cnt$ table is computed recursively using the next-occurrence indices
    \begin{equation*}
        j_0 =\NextPos[i][0],\qquad j_1 = \NextPos[i][1],
    \end{equation*}
    where we recall that $\NextPos[i][b]=n$ indicates that the symbol $b$ does not appear in $x_{[i: n-1]}$. For $t>0$, we have
    \begin{equation*}
        \Cnt[i][t] =
\begin{cases}
0, & \text{if $j_0=n$ and $j_1=n$},\\
\Cnt[j_0+1][t-1], & \text{if $j_0<n$ and $j_1=n$},\\
\Cnt[j_1+1][t-1], & \text{if $j_1<n$ and $j_0=n$},\\
\Cnt[j_0+1][t-1]+\Cnt[j_1+1][t-1], & \text{if $j_0<n$ and $j_1<n$}.
\end{cases}
\end{equation*}

\end{itemize}

\begin{algorithm}
\caption{\textsc{BuildNextPos}($x$)}
\label{alg:buildnextpos}
\KwIn{Binary string $x_{[0: n-1]}$}
\KwOut{$\NextPos[0\ldots n][0\ldots 1]$}
Initialize $\NextPos[i][b]\gets n$ for all $i,b$\;
\For{$i\gets n-1\ \mathrm{downto}\ 0$}{
  $\NextPos[i][0] \gets \NextPos[i+1][0]$\;
  $\NextPos[i][1] \gets \NextPos[i+1][1]$\;
  $\NextPos[i][x[i]] \gets i$\;
}
\Return $\NextPos$
\end{algorithm}

\paragraph{Computational complexity.}
Computing the $\NextPos$ and $\Cnt$ tables requires time and space $O(nk)$.

\subsection{The enumeration algorithm}

We are now in place to describe our unranking algorithm in detail.
The algorithm is described in \cref{alg:unrankk}.
We prove its correctness in this section.

\begin{algorithm}
\caption{$\unrankk(x,k,j)$}
\label{alg:unrankk}
\KwIn{Binary string $x$ of length $n$; $\NextPos$ table; $\Cnt$ table; target length $k$; rank $j$ with $0\le j<\Cnt[0][k]$}
\KwOut{The $j$-th lexicographically smallest distinct subsequence of length $k$}
$i \gets 0$\;
$\rem \gets k$\; 
$\subseq \gets \texttt{empty list}$\; 

\While{$\rem>0$}{
 $j_0 \gets \NextPos[i][0]$\; 
 \tcp{Find the index of the first $0$ at or after position $i$}

 \If{$j_0 < n$}{
   $c_0 \gets \Cnt[j_0+1][\rem-1]$\;
   \tcp{Number of subsequences of x of length $\rem-1$ starting after $j_0$}
 }
 \Else{$c_0 \gets 0$}

 \If{$j < c_0$}{
 append $0$ to $\subseq$\;
 $i \gets j_0+1$\;
 \tcp{Increment index to the position after the chosen $0$}
 $\rem \gets \rem-1$\;
 \textbf{continue}\;
 }

 $j \gets j - c_0$\;

 $j_1 \gets \NextPos[i][1]$\;
 \tcp{Find the index of the first $1$ at or after position $i$}

 \If{$j_1 < n$}{
   $c_1 \gets \Cnt[j_1+1][\rem-1]$\;
   \tcp{Number of subsequences of x of length $\rem-1$ starting after $j_1$}
 }
 \Else{$c_1 \gets 0$}

 \If{$j < c_1$}{
 append $1$ to $\subseq$\;
 $i \gets j_1+1$\;
 \tcp{Increment index to the position after the chosen $1$}
 $\rem \gets \rem-1$\;
 \textbf{continue}\;
 }
}
\Return{$\subseq$}\;
\end{algorithm}

Let $N_{x,k}$ denote the number of length-$k$ subsequences of a string $x$.
The next lemma states, in particular, that the unranking algorithm always returns a length-$k$ subsequence of $x$ when $0\leq j<N_{x,k}$.

\begin{lemma}\label{lem:surj}
If $0\le j<N_{x,k}$, then $\unrankk(x,k,j)$ returns a length-$k$ subsequence of $x$. 
More precisely, at the start of iteration $t$ (i.e., after $t$ bits have been appended), the invariant
\[
0 \le j < \Cnt[i][\rem],\qquad \text{where }\rem=k-t
\]
holds, and $\subseq$ is a subsequence of $x_{[0: i-1]}$.
\end{lemma}

\begin{proof}
We prove this statement by induction.

\paragraph{Initialization.} At the start of the execution of the algorithm we have $t=0$, $\rem=k$, $i=0$, and the precondition requires $0\leq j<N_{x,k}=\Cnt[0][k]$.  
The empty string $\subseq$ is trivially a subsequence of the empty prefix.

\paragraph{Inductive step.}  
Assume the invariant holds at the $t$-th iteration, i.e.,
\[
    0 \le j < \Cnt[i][\rem],
    \quad \subseq \subseteq x_{[0: i-1]},
    \quad \rem=k-t>0.
\]
Let
\[
j_0 = \NextPos[i][0], 
\qquad 
c_0 = \Cnt[j_0+1][\rem-1],
\]
with $c_0=0$ if $j_0=n$.  
Similarly, let
\[
j_1 = \NextPos[i][1], 
\qquad 
c_1 = \Cnt[j_1+1][\rem-1],
\]
with $c_1=0$ if $j_1=n$.  
By the recurrence relation for $\Cnt$ we have
\[
\Cnt[i][\rem] = c_0 + c_1.
\]

We proceed by cases.
\begin{enumerate}
    \item \emph{($j<c_0$)} In this case the algorithm appends $0$ at position $j_0$, sets $i\gets j_0+1$, $\rem\gets\rem-1$, and leaves $j$ unchanged.  
    Since $j<c_0=\Cnt[j_0+1][\rem-1]$, we obtain
    \[
    0 \le j < \Cnt[i][\rem],
    \]
    with the updated $i,\rem$. The new subsequence is valid because it extends by a $0$ occurring at $j_0$.

    \item \emph{($j\ge c_0$)}  
In this case we update $j\gets j-c_0$. From the recurrence
\[
\Cnt[i][\rem]=c_0+c_1,
\]
the condition $j<\Cnt[i][\rem]$ implies $j-c_0<c_1$.  
The algorithm appends $1$ at position $j_1$, sets $i\gets j_1+1$, $\rem\gets\rem-1$.  
Thus the invariant becomes
\[
0 \le j < \Cnt[i][\rem],
\]
with the new $i,\rem$, and the subsequence remains valid.
\end{enumerate}

In both cases the invariant is preserved.

\medskip\noindent
\textbf{Termination.}  
Each loop decreases $\rem$ by one. After $k$ iterations we have $\rem=0$ and exactly $k$ bits appended. By the invariant, $\subseq$ is a subsequence of $x$ of length $k$. Thus the algorithm always returns a valid subsequence of length $k$.
\end{proof}

\cref{lem:surj} shows that the unranking algorithm in \cref{alg:unrankk} always returns a length-$k$ subsequence of $x$.
Therefore, we are done if we prove that running the unranking algorithm with $j\neq j'$ yields distinct subsequences.
To prove this we analyze how the state of the algorithm evolves from one iteration to the next.
More precisely, for fixed $i$ and $r>0$ define the ``transition map''
\[
f_{i,r}:\{0,\dots,\Cnt[i][r]-1\}\longrightarrow
\{0,1\}\times\{0,\dots,n\}\times\{0,\dots,r-1\}\times\N
\]
by
\[
f_{i,r}(j)=
\begin{cases}
(0, \NextPos[i][0]+1,r-1, j), & j<c_0,\\
(1, \NextPos[i][1]+1, r-1, j-c_0), & \text{otherwise,}
\end{cases}
\]
where $c_0=\Cnt[\NextPos[i][0]+1][r-1]$ (with $c_0=0$ if $\NextPos[i][0]=n$).

\begin{lemma}\label{lemma:f_injective}
For fixed $i$ and $r>0$ the map $f_{i,r}$ is injective.
\end{lemma}
\begin{proof}
Fix any $j\neq j'$.
If $f_{i,r}(j)=f_{i,r}(j')$ then, by the last coordinate of the output and the fact that $j\neq j'$, we may assume without loss of generality that $j'=j+c_0$.
But this implies that $j'\geq c_0$, and so $f_{i,r}(j)$ and $f_{i,r}(j')$ differ in the first coordinate.
\end{proof}

\begin{lemma}\label{lem:glob_injk}
Fix $x\in\bits^n$ and $k\in\{1,\dots,n\}$.
Then, the map $j\mapsto \unrankk(x,k,j)$ for $j \in \{0,\dots, \Cnt[0][k]-1\}$ is injective.
\end{lemma}
\begin{proof}
Suppose that $\unrankk(x,k,j)=\unrankk(x,k,j')$ with $0\le j<j'<\Cnt[0][k]$. 
Let $ (i_t,\rem_t,j_t)$ and $ (i'_t,\rem'_t,j'_t)$ denote the states at the start of the $t$-th iteration when running the algorithm with $j$ and $j'$, respectively. 

Since  $\unrankk(x,k,j)=\unrankk(x,k,j')$, the chosen bits, i.e., the bit we add to $\subseq$ at each iteration of the algorithm, are equal.
This means that $i_t=i'_t$ and $\rem_t=\rem'_t$ for all $t$ (the next search index depends only on the previous $i_t$ and the chosen bit). 

We now argue by induction on $t$ that $j_t \neq j'_t$ for all $t$.
For the base case $t=0$, note that $j_0=j<j'_0=j'$ by assumption.
For the induction step, we assume that $j_t\neq j'_t$ and show that $j_{t+1}\neq j'_{t+1}$. 
Because the chosen bit at iteration $t$ is the same in both $\unrankk(x,k,j)$ and $\unrankk(x,k,j')$, and the local transition $f_{i_t,\rem_t}$ is injective by \cref{lemma:f_injective}, the next values satisfy $j_{t+1} \neq j'_{t+1}$. 
By induction, we conclude that $j_k \neq j'_k$.
However, when the algorithm terminates we have $\rem_k=0$ and $\Cnt[i_k][0]=1$, so $j_k=j'_k=0$, a contradiction.
\end{proof}

Combining \cref{lem:surj,lem:glob_injk} immediately yields the following theorem.
\begin{theorem}
For any fixed $x\in\bits^n$ and $k\in\{1,\dots,n\}$, the map
\[
j \mapsto \unrankk(x,k,j)
\]
for $0\le j<\Cnt[0][k]$
is a bijection between $\{0,\dots,N_{x,k}-1\}$ and the set of length-$k$ subsequences of $x$.
\end{theorem}

\subsection{Computational complexity}

We now discuss the complexity of the subsequence enumeration procedure described in \cref{alg:unrankk}.
As mentioned above, pre-computing the $\NextPos$ and $\Cnt$ tables only needs to be done once, and requires time $O(nk)$.
Afterwards, each execution of the $\unrankk$ algorithm takes time $O(k)$, for any rank $j$.
To see this, note that each iteration of the while cycle in \cref{alg:unrankk} takes time $O(1)$ and decreases $\rem$ by at least $1$.
Since $\rem$ is initially set to $k$, the claim follows.

To see the advantage of combining this method with our parallelization of a BAA iteration, recall from \cref{sec:overview} that each thread in the block of the CUDA kernel associated with input $x$ enumerates over $N\approx \frac{N_{x,k}}{1024}$ length-$k$ subsequences of $x$.
By the previous paragraph, the worst-case running time of a thread is $O(nk) + O(Nk)$, with mild hidden constants.
In particular, since we only consider $k\ll 1024$, this worst-case running time improves significantly over a single thread enumerating over $N_{x,k}$ subsequences in time $O(nk+N_{x,k})$.

\section{Enumerating supersequences of a given length}
\label{sec:superseq}

As discussed in \cref{sec:overview}, we also parallelize the computation of the channel output probabilities $Y^{(t)}(y)$ for every output $y\in\bits^k$.
As a sub-routine of this computation, every thread in the block associated to output $y$ in the parallel implementation must enumerate over a certain subset of length-$n$ supersequences of $y$, ordered in some pre-specified way.
We discuss the methods we use to enumerate these supersequences.
As in \cref{sec:subseq}, our methods are tailored to our parallel architecture.

Before we begin, we note that the number of length-$n$ supersequences of $y$ only depends on the length of $y$.
More precisely, the following holds.

\begin{lemma}[\cite{CS75}, for binary strings]
\label{lem:num_superseq}
For any $y\in\bits^k$, the number of length-$n$ supersequences of $y$ is
\begin{equation*}
    \sum_{i = k}^{n} \binom{n}{i} = \sum_{i=0}^{n-k} \binom{n}{i}.
\end{equation*}
\end{lemma}

We divide our enumeration into two parts.
In the first part we enumerate subsets of $\{1,\dots,n\}$ of a given size $0\leq w\leq n-k$.
Intuitively, these subsets represent the coordinates of bits in the supersequence that are not matched to occurrences of $y$ as a subsequence.
In the second part, we show how to map each such subset to a distinct supersequence of $y$.
By \cref{lem:num_superseq} we cover all supersequences of $y$.

\subsection{Enumerating subsets of unmatched bits}\label{sec:enum-subsets}

We begin by analyzing our procedure for enumerating size-$w$ subsets of $\{1,\dots,n\}$.
Recall that $\binom{n}{\leq t}=\sum_{j=0}^t \binom{n}{j}$.
We use the convention that $\binom{n}{\leq -1}=0$.
For each $i \in \left\{0, \ldots, \binom{n}{\leq n-k}-1 \right\}$, let $w \in \{0,1,\ldots,n-k\}$ be the unique integer such that
\begin{equation*}
    \binom{n}{\leq w-1} \leq i <\binom{n}{\leq w}.
\end{equation*}
We then define $j = i - \binom{n}{\leq w-1}$, which satisfies $0 \leq j < \binom{n}{w}$.

\cref{alg:comb} describes the procedure that, given $n$, $w$, and $j$, returns the $j$-th size-$w$ subset of $\{1,\dots,n\}$ according to a pre-specified order. 
Intuitively, this procedure constructs the subset by scanning the ground set $\{1,\dots,n\}$ from left to right and deciding, at each position $t$, whether $t$ is included in the subset. Conceptually, all size-$w$ subsets are partitioned into two classes: those whose smallest remaining element is $t$, and those whose smallest element is larger than $t$. The quantity $\binom{n-t-1}{w-1}$ counts how many subsets fall into the first block (i.e., those that include $t$). We compare the rank $j$ with this number: if $j$ falls within this block, we include $t$ in the subset and continue recursively with the remaining $w-1$ elements chosen from $\{t+1,\dots,n\}$; otherwise, we skip $t$, subtract this block size from $j$, and continue searching among subsets that only use larger elements. Repeating this process reconstructs the $j$-th size-$w$ subset in lexicographic order of \jnote{confirm} the characteristic vectors.\footnote{The characteristic vector of a set $S\subseteq\{1,2,\dots,n\}$ is the length-$n$ binary vector $v_S$ such that $(v_S)_i=1$ if and only if $i\in S$.}

\begin{algorithm}
\caption{$\subunrank(n,w,j)$}
\label{alg:comb}
\KwIn{Integers $n$, $w$, and $j$, where $0 \leq j < \binom{n}{w}$}
\KwOut{List $\texttt{subset}$ containing the $j$-th size-$w$ subset of $\{1,\ldots, n\}$}
\SetKwFunction{Binom}{Binomial}

\BlankLine
Initialize empty list $\texttt{subset} \gets [\,]$\;
\tcp{Recall that indices start at $0$}

\For{$t \gets 0$ \KwTo $n-1$}{
 \If{$w = 0$}{
 \textbf{break}\;
 }

 $c \gets \binom{n - t - 1}{w - 1}$\;
 \tcp{Number of size-$w$ subsets whose smallest element is $t$}

 \If{$j < c$}{
 Append $t$ to $\texttt{subset}$\;
 $w \gets w - 1$\;
 }
 \Else{
 $j \gets j - c$\;
 }
}
\Return{$\texttt{subset}$}\;
\end{algorithm}

It is clear that \cref{alg:comb} returns a size-$w$ subset of $\{1,\dots,n\}$.
It remains to prove that for fixed $n$ and $w$ different $j$'s yield different subsets. 

\begin{lemma}\label{lem:inj-subunrank}
Fix a sequence $y$ of length $k$. For any integers $n,w$ with $0\leq w\leq n-k$ and any
$0 \leq j < j' < \binom{n}{w}$ we have
\begin{equation*}
    \subunrank(n,w,j) \neq \subunrank(n,w,j').
\end{equation*}
\end{lemma}
\begin{proof}
We prove this result by induction on the pair $(w,n)$ under the lexicographic order.

\paragraph{Base case.}
   If $w=0$ then $\binom{n}{0}=1$, so there is no pair $j<j'$.

\paragraph{Inductive step.}  
Fix $(w,n)$ with $1 \le w < n-k$.  Assume the lemma holds for all pairs $(w',n')$ with
$(w',n') < (w,n)$.
Fix also integers $0 \le j < j' < \binom{n}{w}$.  
The $\subunrank(n,w,j)$ procedure first finds the smallest
$t\in\{1,\dots,n\}$ such that
\begin{equation*}
    \sum_{i=0}^{t-1} \binom{n- i -1}{w-1} \leq j < \sum_{i=0}^{t} \binom{n- i -1}{w-1}.
\end{equation*}
Likewise, $\subunrank(n,w,j')$ finds $t'$.
We consider two cases:
\begin{itemize}
    \item If $t\neq t'$, then the subsets returned by $\subunrank(n,w,j)$ and $\subunrank(n,w,j')$ differ in their smallest element.

    \item If $t=t'$, then both $\subunrank(n,w,j)$ and $\subunrank(n,w,j')$ add $t$ to their subsets. 
    Then, they generate a size-$(w-1)$ subset of $\{t+1,\dots,n\}$, exactly like a call to $\subunrank(n-(t-1),w-1,j_{\rem})$ and $\subunrank(n-(t-1),w-1,j'_{\rem})$, respectively,
    where
    \begin{equation*}
        j_{\rem} = j - \sum_{i=0}^{t-1}\binom{n-i-1}{w-1},
    \quad
    j'_{\rem} = j' - \sum_{i=0}^{t-1}\binom{n-i-1}{w-1}.
    \end{equation*}
    Note that $(w-1,n-t-1)<(w,n)$ in lexicographic order.
Hence, by the induction hypothesis, the two calls
produce two distinct $(w-1)$-size subsets of $\{t+1,\dots,n\}$.  
Prepending the same
element $t$ to each yields two distinct $w$-size subsets of $\{1,\dots,n\}$. \qedhere
\end{itemize}
\end{proof}

\subsection{From subsets to supersequences}

We now analyze a procedure that constructs a length-$n$ supersequence of a length-$k$ string $y$ based on the subset output by $\subunrank$. 
The procedure is described in detail \cref{alg:superseq}.
Intuitively, this procedure constructs a length-$n$ supersequence $x$ of $y$ by first deciding whose coordinates of $x$ match the symbols of $y$, and then filling the remaining positions so that they do not interfere with this matching. 
First, we use the procedure from \cref{sec:enum-subsets} to construct a subset $S$ collecting all coordinates of $x$ that are not used to match $y$.
In particular, this means that we must place the bits of $y$ in the coordinates dictated by the complement $\overline{S}$. 
We then scan $x$ from left to right: whenever the current position is not in $S$, we copy the next symbol of $y$; otherwise, we assign the opposite bit so that this position cannot be used to match $y$. In this way, each choice of subset $S$ yields a unique supersequence, and the enumeration of subsets directly induces an enumeration of all supersequences.

\begin{algorithm}
\caption{$\recsuper(i,n,k,y)$}
\label{alg:superseq}
\KwIn{Index $i$ with $0\le i < P(n-k)$; integers $n\ge k$; string $y\in\{0,1\}^k$.}
\KwOut{The $i$-th length-$n$ supersequence $x$ of $y$ according to some pre-specified order.}

Find unique $w\in\{0,\dots,n-k\}$ such that
$\binom{n}{\leq w-1}\le i < \binom{n}{\leq w}$\;
\tcp{Determine the number of positions not used to match symbols of $y$}

Set $j \gets i -\binom{n}{\leq w-1}$\;
\tcp{Index among the $\binom{n}{w}$ choices of unused positions}

$S \gets \subunrank(n,w,j)$\;
\tcp{Compute the $j$-th size-$w$ subset of $\{0,\dots,n-1\}$}

Initialize $x$ as an $n$-bit vector; set $r\gets 0$\;
\tcp{$r$ indexes the next symbol of $y$ to be matched}

\For{$p=0$ to $n-1$}{
 \If{$p\in S$}{
 \eIf{$r < k$}{
   set $x_p \gets 1-y_r$
 }{
   set $x_p \gets 1-y_{k-1}$
 }
 \tcp{Choose a value that does not match the next symbol of $y$}
 }
 {
 \If{$r < k$}{
   set $x_p \gets y_r$\;
   $r\gets r+1$\;
   \tcp{Match the next symbol of $y$}
 }
 \Else{
   set $x_p$ arbitrarily (e.g., $0$)\;
 }
 }
}
\Return $x$\;
\end{algorithm}

We begin by proving that $\recsuper$ always outputs a length-$n$ supersequence of $y$.
\begin{lemma}\label{lem:correctness}
For every $i$ with $0\le i < \binom{n}{\leq n-k}$, the output $x=\recsuper(i,n,k,y)$ is a supersequence of $y$.
\end{lemma}
\begin{proof}
Let $w$ and $M$ be the integer and subset computed in a call to $\recsuper(i,n,k,y)$.
By construction, $M$ has size $w\le n-k$.
In the for loop we increment $r$ only when $p\not\in M$ and $r<k$, in which case we match the next bit of $y$ to the next bit of $x$.
Since the complement of $M$ has size $n-w\geq k$, when we exit the for loop we have $r=k$, and so we match all bits of $y$ to bits of $x$.
\end{proof}

Now we prove that calling $\recsuper$ on distinct $i$'s yields distinct supersequences of $y$.
\begin{lemma}
\label{lem:decode_superseq}
If $i\neq i'$, then $\recsuper(i,n,k,y)\neq\recsuper(i',n,k,y)$.
\end{lemma}
\begin{proof}
Let $(w,S,j)$ be the integers and subset first computed by $\recsuper(i,n,k,y)$, and $(w',S',j')$ the integer and subset first computed by $\recsuper(i',n,k,y)$.
Denote the supersequences by these calls by $x$ and $x'$, respectively.

There are two cases to consider.
If $w\neq w'$ then $|S|=w\neq w'=|S'|$, and so $S\neq S'$ in particular.
If $w=w'$ then $j\neq j'$, and so $S=\subunrank(n,w,j)\neq \subunrank(n,w,j')=S'$ by \cref{lem:inj-subunrank}.
Therefore, it is always the case that $S\neq S'$.

Let $\ell$ be the smallest index at which $S$ and $S'$ differ.
By construction, this means that $S\cap\{0,\dots,\ell-1\}=S'\cap\{0,\dots,\ell-1\}$.
Therefore, the construction of the supersequences $x$ and $x'$ up to position $\ell-1$ is identical.
In particular, both have consumed the same number $c$ of symbols from $y$ by position $\ell$, where
\begin{equation*}
    c = \left|\{u<j: u\not\in S\}\right|.
\end{equation*}
At index $\ell$ the two calls behave differently.
Without loss of generality, assume that $\ell\in S$ and $\ell\not\in S'$.
Then,
\begin{equation*}
    x_\ell = 1 - y_c,\quad x'_\ell = y_c,
\end{equation*}
and so $x_\ell\neq x'_\ell$.
\end{proof}

Combining \cref{lem:num_superseq,lem:correctness,lem:decode_superseq} immediately yields the following result.
\begin{theorem}
For any $y\in\{0,1\}^k$ and $n\ge k$ the map
\[
i \mapsto \recsuper(i,n,k,y)
\]
is a bijection between $\{0,\dots,S_{k,n}-1\}$ and the set of length-$n$ supersequences of $y$.
\end{theorem}

\subsection{Computational complexity}

We now discuss the complexity of the supersequence enumeration procedure described in \cref{alg:superseq}.
Each call to $\subunrank$ takes $O(n)$ time.
The remainder of \cref{alg:superseq} also runs in $O(n)$ time, and so, overall, this algorithm runs in $O(n)$ time.

As in \cref{sec:subseq}, we gain an advantage by combining this method with our parallelization of a BAA iteration.
Each thread in a block associated with output $y\in\bits^k$ enumerates over $N\approx \frac{S_{k,n}}{1024}$ length-$n$ supersequences of $x$.
By the previous paragraph, the running time of a thread is $O(N\cdot n)$, with a mild hidden constant.
Since $n\ll 1024$ in our setting, this improves significantly over a single thread enumerating over $S_{k,n}$ supersequences.

\section{Results}
\label{sec:results}
 
Using our optimized parallelized implementation of the BAA, we managed to compute good upper bounds on $C_{n,k}$ for all pairs $(n,k)$ with $k \le n \leq 29$.
Going even further, we also managed to compute good upper bounds on $C_{31,k}$ for all $k\leq 18$.
For most of the upper bound computations a tolerance of $a=0.005$ was enforced, ensuring that the true value of $C_{n,k}$ exceeds the information rate returned by the BAA by at most $0.005$ (see \cref{claim:baa} and the surrounding discussion in \cref{sec:baa}). 
The only exceptions are in the approximation of $C_{31,k}$ for $14 \leq k \leq 18$, where we enforced a tolerance of $a = 0.05$ due to the rapid increase in time per iteration, reaching 2500 seconds per iteration for $k = 18$.
Reported capacity upper bounds are obtained by adding the respective tolerance value to the rate output by the BAA.

All computations were done with an \textit{RTX 5070 Ti} GPU. For $n = 29$, the computations took at most 1100 iterations to complete, and for $k \approx \frac{n}{2}$ each iteration took approximately 400 seconds, taking 5 days to complete. For $n = 31$, the computations for $k \leq 13$ took less then 500 seconds to complete, taking up to one week to complete, especially for $k = 13$. For $14 \leq k \leq 18$, each iteration took more than 800 seconds to complete, hence the need for a higher tolerance.

The upper bounds on $C_{29,k}$  for all $k\in\{1,\dots,28\}$ are reported in \cref{tab:C_29_k}. 
The upper bounds we obtained on $C_{31,k}$ are reported in \cref{tab:C_31_k}.
To upper bound the $C_{31,k}$ values
for $k>18$ we combine our upper bounds on $C_{n,k}$ for $n\leq 29$ and the following known lemma that upper bounds $C_{n,k}$ based on $C_{n',k'}$ for $n'<n$.

\begin{lemma}[\cite{RC23}]\label{lem:worse-UB}
    For every $s \in \{1,\ldots,n\}$ we have that 
    \begin{equation*}
        C_{n,k} \leq \sum_{i=0}^s \frac{\binom{s}{i}\binom{n-s}{k-i}}{\binom{n}{k}} \left( C_{s,i} + C_{n-s,k-i} \right).
    \end{equation*}
\end{lemma}
More precisely, we set $s=2$ and combine our upper bounds on $C_{29,i}$ reported in \cref{tab:C_29_k} with the easy-to-compute values $C_{2,1}=1$ and $C_{2,2}=2$ for all relevant $i$ and $j$.

We use our upper bounds on the $C_{29,k}$ and $C_{31,k}$ values to obtain upper bounds on the capacity $C(d)$ of the binary deletion channel with deletion probability $d$ via \cref{lem:prob-vs-exact}.
These bounds are reported in \cref{tab:C_d}, where they are also compared to the previous best known upper bounds~\cite{RC23}.
Note that since our upper bounds on $C_{31,k}$ are loose for $k>18$, the upper bounds on $C(d)$ obtained by plugging our upper bounds on $C_{29,k}$ into \cref{lem:prob-vs-exact} are better than those obtained via our upper bounds on $C_{31,k}$ when $d$ is not large.

\paragraph{Capacity upper bounds in the high-noise regime.}
Our improved upper bounds on $C(d)$ also lead to an improved upper bound on $C(d)$ in the asymptotic regime $d\to 1$.
This is obtained via the following result due to Rahmati and Duman~\cite{RD15}.

\begin{lemma}[{\cite{RD15}}]\label{lem:convex}
Let $\lambda, d' \in [0,1]$ and define $d = \lambda d' + 1 - \lambda$. Then,
\[
\frac{C(d)}{1-d} \;\leq\; \frac{C(d')}{1-d'}.
\]
\end{lemma}

When $d' = 0.64$, our new upper bound yields
\[
    \frac{C(d')}{1-d'} \leq 0.3578.
\]
Therefore, instantiating \cref{lem:convex} with this upper bound yields
\[
   C(d) \;\leq\; 0.3578 (1-d)
\]
for all $d \geq 0.64$.

\section*{Acknowledgments}

We thank Roni Con for insightful discussions and detailed feedback on earlier versions of this work.

\bibliographystyle{alpha}
\bibliography{refs}

@article{Mit09,
  author    = {Michael Mitzenmacher},
  title     = {A Survey of Results for Deletion Channels and Related Synchronization Channels},
  journal   = {Probability Surveys},
  volume    = {6},
  pages     = {1--33},
  year      = {2009}
}

@article{CR20,
  author    = {Mahdi Cheraghchi and Jo{\~a}o Ribeiro},
  title     = {An Overview of Capacity Results for Synchronization Channels},
  journal   = {IEEE Transactions on Information Theory},
  volume    = {67},
  number    = {6},
  pages     = {3207--3232},
  year      = {2021}
}

@article{FD10,
  author    = {Dario Fertonani and Tolga M. Duman},
  title     = {Novel Bounds on the Capacity of the Binary Deletion Channel},
  journal   = {IEEE Transactions on Information Theory},
  volume    = {56},
  number    = {6},
  pages     = {2753--2765},
  year      = {2010}
}

@inproceedings{Dal11,
  author    = {Marco Dalai},
  title     = {A New Bound on the Capacity of the Binary Deletion Channel with High Deletion Probabilities},
  booktitle = {2011 IEEE International Symposium on Information Theory (ISIT)},
  pages     = {499--502},
  year      = {2011},
  publisher = {IEEE}
}

@INPROCEEDINGS{RC23,
  author={Rubinstein, Ittai and Con, Roni},
  booktitle={2023 IEEE International Symposium on Information Theory (ISIT)}, 
  title={Improved Upper and Lower Bounds on the Capacity of the Binary Deletion Channel}, 
  year={2023},
  volume={},
  number={},
  pages={927-932},
  doi={10.1109/ISIT54713.2023.10206626}}

@article{Bla72,
  author    = {Richard E. Blahut},
  title     = {Computation of Channel Capacity and Rate-Distortion Functions},
  journal   = {IEEE Transactions on Information Theory},
  volume    = {18},
  number    = {4},
  pages     = {460--473},
  year      = {1972}
}

@article{Ari72,
  author    = {Suguru Arimoto},
  title     = {An Algorithm for Computing the Capacity of Arbitrary Discrete Memoryless Channels},
  journal   = {IEEE Transactions on Information Theory},
  volume    = {18},
  number    = {1},
  pages     = {14--20},
  year      = {1972}
}

@article{CS75,
  author    = {V. Chvátal and D. Sankoff},
  title     = {Longest Common Subsequence of Two Random Sequences},
  journal   = {Journal of Applied Probability},
  number    = {12},
  pages     = {306--315},
  year      = {1975}
}

@article{Sha48,
	author = {Claude E. {Shannon}},
	journal = {The Bell System Technical Journal},
	title = {A mathematical theory of communication},
volume = {27},
	number={3},
	pages={379-423},
	year = {1948},
    doi={10.1002/j.1538-7305.1948.tb01338.x}
}

@ARTICLE{HS21,
  author={Haeupler, Bernhard and Shahrasbi, Amirbehshad},
  journal={IEEE Transactions on Information Theory}, 
  title={Synchronization strings and codes for insertions and deletions—A survey}, 
  year={2021},
  volume={67},
  number={6},
  pages={3190-3206},
  doi={10.1109/TIT.2021.3056317}}

@ARTICLE{MBT10, 
author={Hugues {Mercier} and Vijay K. {Bhargava} and Vahid {Tarokh}}, 
journal={IEEE Communications Surveys Tutorials}, 
title={A survey of error-correcting codes for channels with symbol synchronization errors}, 
year={2010}, 
volume={12}, 
number={1}, 
pages={87-96}, 
doi={10.1109/SURV.2010.020110.00079}, 
ISSN={1553-877X}, 
month={First Quarter},
}

@INPROCEEDINGS{DK07,
author={Eleni Drinea and Adam Kirsch},
booktitle={2007 IEEE International Symposium on Information Theory (ISIT)},
title={Directly lower bounding the information capacity for channels with i.i.d. deletions and duplications},
year={2007},
volume={},
number={},
pages={1731-1735},
doi={10.1109/ISIT.2007.4557471},
ISSN={2157-8095},
}

@inproceedings{DMP07,
  title={Capacity upper bounds for the deletion channel},
  author={Diggavi, Suhas and Mitzenmacher, Michael and Pfister, Henry D.},
  booktitle={2007 IEEE International Symposium on Information Theory (ISIT)},
  pages={1716--1720},
  year={2007},
    doi={10.1109/ISIT.2007.4557469}
}

@article{RD15,
  title={Upper bounds on the capacity of deletion channels using channel fragmentation},
  author={Rahmati, Mojtaba and Duman, Tolga M.},
  journal={IEEE Transactions on Information Theory},
  volume={61},
  number={1},
  pages={146--156},
  year={2015},
  publisher={IEEE},
    doi={10.1109/TIT.2014.2368553}
}

@ARTICLE{RA13,
author={Mahdi Ramezani and Masoud Ardakani},
journal={IEEE Transactions on Communications},
title={On the capacity of duplication channels},
year={2013},
volume={61},
number={3},
pages={1020-1027},
doi={10.1109/TCOMM.2013.020413.120070},
}

@article{Mit08,
	title={Capacity bounds for sticky channels},
	author={Mitzenmacher, Michael},
	journal={IEEE Transactions on Information Theory},
	volume={54},
	number={1},
	pages={72--77},
	year={2008},
	publisher={IEEE},
    doi={10.1109/TIT.2007.911291}
}

@ARTICLE{MTL12,
author={Hugues Mercier and Vahid Tarokh and Fabrice Labeau},
journal={IEEE Transactions on Information Theory},
title={Bounds on the capacity of discrete memoryless channels corrupted by synchronization and substitution errors},
year={2012},
volume={58},
number={7},
pages={4306-4330},
doi={10.1109/TIT.2012.2191682},
}

@article{ISW16,
	title={On the capacity of channels with timing synchronization errors},
	author={Iyengar, Aravind R. and Siegel, Paul H. and Wolf, Jack Keil},
	journal={IEEE Transactions on Information Theory},
	volume={62},
	number={2},
	pages={793--810},
	year={2016},
	publisher={IEEE},
    doi={10.1109/TIT.2015.2504358}
}

@INPROCEEDINGS{CK15, 
author={Jason {Castiglione} and Aleksandar {Kav\v{c}i\'{c}}}, 
booktitle={2015 IEEE Information Theory Workshop - Fall (ITW)}, 
title={Trellis based lower bounds on capacities of channels with synchronization errors}, 
year={2015}, 
volume={}, 
number={}, 
pages={24-28}, 
doi={10.1109/ITWF.2015.7360727}, 
ISSN={}, 
month={Oct},}

@techreport{Gal61,
  title={Sequential decoding for binary channels with noise and synchronization errors},
  author={Gallager, Robert G.},
  year={1961},
  institution={MIT Lexington Lincoln Laboratory}
}

@article{Zig69,
  title={Sequential decoding for a binary channel with drop-outs and insertions},
  author={Zigangirov, Kamil Sh.},
  journal={Problemy Peredachi Informatsii},
  volume={5},
  number={2},
  pages={23--30},
  year={1969},
  publisher={Russian Academy of Sciences, Branch of Informatics, Computer Equipment and~…},
}

@article{VD68,
  title={The computation on a computer of the channel capacity of a line with symbol drop-out},
  author={Vvedenskaya, Nikita D. and Dobrushin, Roland L.},
  journal={Problemy Peredachi Informatsii},
  volume={4},
  number={3},
  pages={92--95},
  year={1968},
  publisher={Russian Academy of Sciences, Branch of Informatics, Computer Equipment and~…},
    url={https://www.mathnet.ru/rus/ppi1865}
}

@ARTICLE{DG06, 
author={Suhas {Diggavi} and Matthias {Grossglauser}}, 
journal={IEEE Transactions on Information Theory}, 
title={On information transmission over a finite buffer channel}, 
year={2006}, 
volume={52}, 
number={3}, 
pages={1226-1237}, 
doi={10.1109/TIT.2005.864445},  
month={March},}

@article{DM07,
	title={Improved lower bounds for the capacity of iid deletion and duplication channels},
	author={Drinea, Eleni and Mitzenmacher, Michael},
	journal={IEEE Transactions on Information Theory},
	volume={53},
	number={8},
	pages={2693--2714},
	year={2007},
	publisher={IEEE},
    doi={10.1109/TIT.2007.901221}
}

@article{VTR13,
	title={Achievable rates for channels with deletions and insertions},
	author={Venkataramanan, Ramji and Tatikonda, Sekhar and Ramchandran, Kannan},
	journal={IEEE Transactions on Information Theory},
	volume={59},
	number={11},
	pages={6990--7013},
	year={2013},
	publisher={IEEE},
    doi={10.1109/TIT.2013.2278181}
}

@article{Che19,
 author = {Cheraghchi, Mahdi},
 title = {Capacity upper bounds for deletion-type channels},
 journal = {J. ACM},
 issue_date = {April 2019},
 volume = {66},
 number = {2},
 month = mar,
 year = {2019},
 pages = {9:1--9:79},
 articleno = {9},
 numpages = {79},
 doi = {10.1145/3281275},
 acmid = {3281275},
 publisher = {ACM},
 address = {New York, NY, USA},
}

@ARTICLE{CR19b, 
author={Mahdi {Cheraghchi} and Jo{\~a}o {Ribeiro}}, 
journal={IEEE Transactions on Information Theory}, 
title={Sharp analytical capacity upper bounds for sticky and related channels}, 
year={2019}, 
volume={65}, 
number={11}, 
pages={6950-6974}, 
doi={10.1109/TIT.2019.2920375}, 
month={Nov},}

@article{MD06,
	title={A simple lower bound for the capacity of the deletion channel},
	author={Mitzenmacher, Michael and Drinea, Eleni},
	journal={IEEE Transactions on Information Theory},
	volume={52},
	number={10},
	pages={4657--4660},
	year={2006},
	publisher={IEEE},
    doi={10.1109/TIT.2006.881844}
}

@ARTICLE{KD10, 
author={Adam {Kirsch} and Eleni {Drinea}}, 
journal={IEEE Transactions on Information Theory}, 
title={Directly lower bounding the information capacity for channels with i.i.d. deletions and duplications}, 
year={2010}, 
volume={56}, 
number={1}, 
pages={86-102}, 
doi={10.1109/TIT.2009.2034883}, 
ISSN={}, 
month={Jan},}

@article{YGM17,
	title={Portable and error-free {DNA}-based data storage},
	author={Yazdi, S. M. Hossein Tabatabaei and Gabrys, Ryan and Milenkovic, Olgica},
	journal={Scientific reports},
	volume={7},
	number={1},
	pages={5011},
	year={2017},
	publisher={Nature Publishing Group},
    doi={10.1038/s41598-017-05188-1}
}

@article{OAC+18,
	title={Random access in large-scale {DNA} data storage},
	author={Organick, Lee and Ang, Siena Dumas and Chen, Yuan-Jyue and Lopez, Randolph and Yekhanin, Sergey and Makarychev, Konstantin and Racz, Miklos Z and Kamath, Govinda and Gopalan, Parikshit and Nguyen, Bichlien and others},
	journal={Nature Biotechnology},
	volume={36},
	number={3},
	pages={242},
	year={2018},
	publisher={Nature Publishing Group},
    doi={10.1038/nbt.4079}
}

@article{HMG19,
  title={A characterization of the DNA data storage channel},
  author={Heckel, Reinhard and Mikutis, Gediminas and Grass, Robert N.},
  journal={Scientific reports},
  volume={9},
  number={1},
  pages={9663},
  year={2019},
  publisher={Nature Publishing Group UK London},
    doi={10.1038/s41598-019-45832-6}
}

@INPROCEEDINGS{WGGH23,
  author={Weindel, Franziska and Gimpel, Andreas L. and Grass, Robert N. and Heckel, Reinhard},
  booktitle={2023 59th Annual Allerton Conference on Communication, Control, and Computing (Allerton)}, 
  title={Embracing errors is more effective than avoiding them through constrained coding for {DNA} data storage}, 
  year={2023},
  volume={},
  number={},
  pages={1-8},
  doi={10.1109/Allerton58177.2023.10313494}}

@article{Dob67,
  title={Shannon's theorems for channels with synchronization errors},
  author={Dobrushin, Roland L.},
  journal={Problemy Peredachi Informatsii},
  volume={3},
  number={4},
  pages={18--36},
  year={1967},
  publisher={Russian Academy of Sciences, Branch of Informatics, Computer Equipment and~…},
    url={https://www.mathnet.ru/eng/ppi1919}
}

@book{Yeu08,
  title={Information Theory and Network Coding},
  author={Yeung, Raymond W.},
  year={2008},
  publisher={Springer Science \& Business Media}
}

@misc{repo,
  author = {Martim Pinto},
  title = {{GPU} code for computing capacity upper bounds for the binary deletion channel},
  year = {2026},
  howpublished = {\url{https://doi.org/10.5281/zenodo.19453272}},
}

\newpage

\begin{table}
    \begin{minipage}[t]{.5\linewidth}
        \centering
        \footnotesize
        \caption{Upper bounds on $C_{29,k}$.
        }
        \label{tab:C_29_k}
        \begin{tabular}{c|c}
        \hline
        $k$ & Upper bound on $C_{29,k}$ \\
        \hline
        1  & 1.0000 \\
        2  & 1.1898 \\
        3  & 1.5165 \\
        4  & 1.8137 \\
        5  & 2.0916 \\
        6  & 2.3761 \\
        7  & 2.6647 \\
        8  & 2.9598 \\
        9  & 3.2663 \\
        10 & 3.5867 \\
        11 & 3.9247 \\
        12 & 4.2841 \\
        13 & 4.6727 \\
        14 & 5.0930 \\
        15 & 5.5573 \\
        16 & 6.0713 \\
        17 & 6.6464 \\
        18 & 7.2964 \\
        19 & 8.0364 \\
        20 & 8.8850 \\
        21 & 9.8659 \\
        22 & 11.0096 \\
        23 & 12.3545 \\
        24 & 13.9487 \\
        25 & 15.8526 \\
        26 & 18.1454 \\
        27 & 20.9387 \\
        28 & 24.4132 \\
        29 & 29.0000 \\
        \hline
    \end{tabular}
\end{minipage}
    \begin{minipage}[t]{.5\linewidth}
        \centering
        \caption{
        Upper bounds on $C_{31,k}$.}
        \label{tab:C_31_k}
        \footnotesize
        \begin{tabular}{c|c}
        \hline
        $k$ & Upper bound on $C_{31,k}$ \\
        \hline
        1  & 1.0000 \\
        2  & 1.1888 \\
        3  & 1.5136 \\
        4  & 1.8086 \\
        5  & 2.0831 \\
        6  & 2.3632 \\
        7  & 2.6458 \\
        8  & 2.9332 \\
        9  & 3.2298 \\
        10 & 3.5384 \\
        11 & 3.8605 \\
        12 & 4.2005 \\
        13 & 4.5623 \\
        14 & 4.9511 \\
        15 & 5.3895 \\
        16 & 5.8826 \\
        17 & 6.3947 \\
        18 & 6.9592 \\
        19 & 8.3882 \\
        20 & 9.1247 \\
        21 & 9.9515 \\
        22 & 10.8865 \\
        23 & 11.9522 \\
        24 & 13.1766 \\
        25 & 14.5945 \\
        26 & 16.2486 \\
        27 & 18.1927 \\
        28 & 20.4969 \\
        29 & 23.2603 \\
        30 & 30.0000 \\
        31 & 31.0000 \\
        \hline
        \end{tabular}
    \end{minipage}

\end{table}

\begin{table}[]
    \centering
    \caption{
    Upper bounds on $C(d)$ obtained by combining \cref{lem:prob-vs-exact} with the upper bounds on $C_{29,k}$ from \cref{tab:C_29_k} or the upper bounds on $C_{31,k}$ from \cref{tab:C_31_k}, and taking the minimum between these two.
    For each value of $d$, we list which $n$ ($29$ or $31$) gave the best upper bound.
    We also list the previous best known upper bounds~\cite{RC23} for comparison.}
    \label{tab:C_d}
    \begin{tabular}{c|c|c}
        \hline
        $d$ & New upper bound& Previous upper bound~\cite{RC23} \\
        \hline
        0.01 & 0.9557 ($n=29$) & 0.9583 \\
        0.02 & 0.9141 ($n=29$) & 0.9189 \\
        0.03 & 0.8751 ($n=29$) & 0.8817 \\
        0.04 & 0.8385 ($n=29$) & 0.8467 \\
        0.05 & 0.8039 ($n=29$) & 0.8139 \\
        0.10 & 0.6577 ($n=29$) & 0.6762 \\
        0.15 & 0.5454 ($n=29$) & 0.5660 \\
        0.20 & 0.4574 ($n=29$) & 0.4786 \\
        0.25 & 0.3876 ($n=29$) & 0.4083 \\
        0.30 & 0.3314 ($n=29$) & 0.3513 \\
        0.35 & 0.2857 ($n=29$) & 0.3045  \\
        0.40 & 0.2480 ($n=29$) & 0.2648   \\
        0.45 & 0.2164 ($n=29$) &  0.2309   \\
        0.50 & 0.1896 ($n=29$) &   0.2015   \\
        0.55 & 0.1652 ($n=31$) &    0.1755   \\
        0.60 & 0.1438 ($n=31$) &    0.1524   \\
        0.64 & 0.1288 ($n=31$) & --\\
        0.65 & 0.1253 ($n=31$) &     0.1313   \\
        0.68 & 0.1151 ($n=31$) &      0.1199   \\

        \hline
    \end{tabular}
\end{table}

\newpage

\appendix

\section{Computing channel output probabilities using subsequences}\label{app:output-seqs}

In this section we present an alternative method to compute the channel output probabilities $Y^{(t)}(y)$ of the $\BDC_{n,k}$ in an iteration of the BAA using subsequence enumeration and exploiting symmetries of the $\BDC_{n,k}$ and the input distributions obtained through the BAA.
We found this method to be faster in practice than the method discussed in \cref{sec:superseq} for approximating $C_{n,k}$ using the BAA when $k$ is close to $n/2$.
We note that the idea of using channel symmetries to speed up computations is not new:
it was used before in the optimized implementation of the BAA in~\cite{RC23}, although it was not analyzed in the paper itself.
We exploit these symmetries in a somewhat different way, and we provide an analysis for completeness.

We begin by providing some intuition behind the method.
Recall that in \cref{sec:overview} we computed $Y^{(t)}(y)$ for all $y\in\bits^k$ in parallel by assigning each block in the CUDA kernel to a different output $y\in\bits^k$, and then having threads within that block enumerate over appropriate subsets of length-$n$ supersequences of $y$ using the method from \cref{sec:superseq}.
Alternatively, we can also compute $Y^{(t)}(y)$ for all $y\in\bits^k$ by maintaining an array (with all entries initialized to $0$) storing $Y^{(t)}(y)$ for all $y$.  
Then,
instead we assign each block in the CUDA kernel to a different \emph{input} $x\in\bits^n$, and have threads within that block enumerate over appropriate subsets of length-$k$ subsequences of $x$ using the method from \cref{sec:superseq}, updating the array storing $Y^{(t)}$ as they go along.

This approach can be sped up by taking into account some simple symmetries of the $\BDC_{n,k}$.
In more detail, for an arbitrary integer $n\geq 1$ and a string $x\in\bits^n$ let $\cpl(x)$ denote its coordinate-wise complement (so that $\cpl(x)_i=1-x_i$) and $\rev(x)$ denote its reversal (so that $\rev(x)_i=x_{n-1-i}$), where we write $x=(x_0,\dots,x_{n-1})$.
Note that these functions are their own inverses.
Then, for $g=\cpl$ or $g=\rev$ and any $x\in\bits^n$ and $y\in\bits^k$, we have
\begin{equation}\label{eq:BDC-sym}
    P_{n,k}(g(y)|g(x))=P_{n,k}(y|x).
\end{equation}
The same extends directly to the composition $\cpl\circ\rev$.
Using \cref{eq:BDC-sym}, we can derive analogous symmetries of the input and output distributions produced by the various iterations of the BAA.
\begin{theorem}
\label{thm:baa-symmetry}
Suppose the BAA applied to the $\BDC_{n,k}$ is initialized with a uniform input distribution $X^{(0)}$.
Then, for every $t\geq 1$, $x\in\bits^n$, $y\in\bits^k$, and $g=\cpl$ or $g=\rev$ we have
\begin{equation*}
    X^{(t)}(g(x))=X^{(t)}(x)
\end{equation*}
and
\begin{equation*}
    Y^{(t)}(g(y))=Y^{(t)}(y).
\end{equation*}
\end{theorem}
\begin{proof}
    We established the desired statement by induction in $t$.
    For brevity, we write $P=P_{n,k}$.
    
    The base case $t=0$ is clear since $X^{(0)}$ is uniform over $\bits^n$.
    Now fix $t\geq 1$ and suppose that $X^{(t)}(g(x))=X^{(t)}(x)$ for all $x\in\bits^n$.
    We show that the same thing holds for $X^{(t+1)}$.
    First, note that
    \begin{align}
        Y^{(t)}(g(y)) &= \sum_{x\in\bits^n}X^{(t)}(x) P(g(y)|x)\nonumber\\
        &= \sum_{x\in\bits^n}X^{(t)}(g(x)) P(y|g(x))\nonumber\\
        &= \sum_{x'\in\bits^n}X^{(t)}(x') P(y|x')\nonumber\\
        &=Y^{(t)}(y)\label{eq:inv-Y}
    \end{align}
    for all $y\in\bits^k$.
    The second equality uses the induction hypothesis.
    The third equality uses the fact that $g$ is a bijection.
    Then,
    \begin{align*}
        W^{(t)}(g(x)) &= \prod_{y\in\bits^k}\left(\frac{X^{(t)}(g(x)) P(y|g(x))}{Y^{(t)}(y)}\right)^{P(y|g(x))}\\
        &=\prod_{y\in\bits^k}\left(\frac{X^{(t)}(x) P(g(y)|x)}{Y^{(t)}(g(y))}\right)^{P(g(y)|x)}\\
        &=\prod_{y'\in\bits^k}\left(\frac{X^{(t)}(x) P(y'|x)}{Y^{(t)}(y')}\right)^{P(y'|x)}\\
        &=W^{(t)}(x).
    \end{align*}
    The second equality uses the induction hypothesis and \cref{eq:inv-Y}.
    Since $X^{(t+1)}$ is obtained by normalizing $W^{(t)}$ the desired result follows.
    \qedhere
\end{proof}

We can use \cref{thm:baa-symmetry} to slightly simplify the computation of $(Y^{(t)}(y))_{y\in\bits^k}$.
Given a string $y\in\bits^k$, we define its orbit $\cO_y=\{y,\cpl(y),\rev(y),\rev\circ \cpl(y)\}$.
To each orbit $\cO_y$ we associate as its \emph{representative} the smallest string $y'\in\cO_y$ with respect to the lexicographic order, and denote it by $\rep(y)$.
We denote the set of all representatives in $\bits^k$ by $\cR_k$.

By \cref{thm:baa-symmetry}, it suffices to compute $Y^{(t)}(r)$ for all representatives $r\in\cR_k$.
Furthermore, we have
\begin{align*}
    Y^{(t)}(r) &= \sum_{x\in\bits^n} X^{(t)}(x) P_{n,k}(r|x)\\
    &= \sum_{x\in \cR_n} X^{(t)}(x) \sum_{x'\in \cO_x} P_{n,k}(r|x')\\
    &=\sum_{x\in \cR_n} X^{(t)}(x)\cdot \frac{|\cO_x|}{|\cO_r|} \sum_{y\in\cO_r}P_{n,k}(y|x).
\end{align*}
The second equality uses \cref{thm:baa-symmetry}.
To prove the last equality, we can, for example, use a group-theoretic argument.
Let $G$ be the group generated by $\cpl$ and $\rev$ via composition of functions.
For a binary string $z$, define the \emph{stabilizer} $\Stab_z=\{g\in G: g(z)=z\}$.
Note that $g^{-1}=g$ for every $g\in G$, and that $\cO_z$ is the orbit of $z$ under the action of $G$.
Then,
\begin{align*}
    \sum_{x'\in \cO_x} P_{n,k}(r|x') &= \frac{1}{|\Stab_x|}\sum_{g\in G} P_{n,k}(r|g(x))\\
    &= \frac{1}{|\Stab_x|}\sum_{g\in G} P_{n,k}(g(r)|x)\\
    &= \frac{|\Stab_r|}{|\Stab_x|}\sum_{y\in\cO_r} P_{n,k}(y|x)\\
    &= \frac{|\cO_x|}{|\cO_r|} \sum_{y\in\cO_r}P_{n,k}(y|x).
\end{align*}
The second equality uses the fact that $P_{n,k}(r|g(x))=P_{n,k}(g^{-1}(r)|x)=P_{n,k}(g(r)|x)$ for any $r$ and $x$, and the last equality uses the fact that $|\Stab_z|=|G|/|\cO_z|$ for any $z$.

The discussion above motivates the procedure described in \cref{alg:orbit}.

\begin{algorithm}
\caption{Computing $Y^{(t)}(y)$ for all $y\in\bits^k$}
\label{alg:orbit}
\SetKwInOut{Input}{Input}
\SetKwInOut{Output}{Output}

\Input{Input size $n$, output size $k$, input distribution $X^{(t)}$}
\Output{Array $Y^{(t)}$ indexed by orbit representatives of outputs}

Initialize $Y^{(t)}(r) \leftarrow 0$ for all orbit representatives $r\in\cR_k$

\ForEach{$x \in \cR_n$}{
    $o_x \leftarrow |\cO_x|$

    \ForEach{\textnormal{subsequence $y$ of $x$}}{
        $o_y \leftarrow |\cO_y|$ 
        
        $r \leftarrow \rep(y)$
        
        $Y^{(t)}(r) \leftarrow Y^{(t)}(r) + \dfrac{o_x}{o_y} \cdot X^{(t)}(x) \cdot P_{n,k}(y | x)$
    }
}
\Return{$Y^{(t)}$}
\end{algorithm}

\end{document}